\documentclass[a4paper,12pt]{article}
\usepackage{jheppub}
\usepackage{amsmath,latexsym,amssymb,amsthm,amstext} 
\usepackage{mathbbol}

\newtheorem{prop}{Proposition}
\newtheorem{theo}{Theorem}

\newcommand{\ba}{\begin{eqnarray}}
\newcommand{\ea}{\end{eqnarray}}
\newcommand{\be}{\begin{equation}}
\newcommand{\ee}{\end{equation}}
\newcommand{\gv}{E}
\newcommand{\fv}{L}
\newcommand{\fmoinsunvierbein}{\ell}
\newcommand{\fvm}{\ell}

\title{A note on ``symmetric" vielbeins in bimetric, massive,  perturbative and non perturbative gravities}

\author[a]{C.~Deffayet}
\author[a]{J.~Mourad}
\author[a]{G.~Zahariade}

\affiliation[a]{APC\;\footnote{UMR 7164 (CNRS, Universit\'e Paris 7, CEA, Observatoire de Paris)}, 10 rue Alice Domon et L\'eonie Duquet,
 75205 Paris Cedex 13, France}

\emailAdd{deffayet@iap.fr}
\emailAdd{mourad@apc.univ-paris7.fr}
\emailAdd{zahariad@apc.univ-paris7.fr}

\abstract{
We consider a manifold  endowed with two different vielbeins $\gv^A_{\hphantom{A}\mu}$ and $\fv^A_{\hphantom{A}\mu}$ corresponding to two different  metrics $g_{\mu \nu}$ and $f_{\mu \nu}$. Such a situation arises generically in bimetric or massive gravity 
(including the recently discussed version of de Rham, Gabadadze and Tolley), as well as in perturbative quantum gravity where one vielbein parametrizes the background space-time and the other the dynamical degrees of freedom. 
We determine the conditions under which the relation $g^{\mu \nu} \gv^A_{\hphantom{A}\mu} \fv^B_{\hphantom{A}\nu} = g^{\mu \nu} \gv^B_{\hphantom{A}\mu} \fv^A_{\hphantom{A}\nu}$ can be imposed (or the ``Deser-van Nieuwenhuizen" gauge chosen). We clarify and correct various statements which have been made about this issue. We show in particular that in $D=4$ dimensions, this condition is always equivalent to the existence of a real matrix square root of $g^{-1}f$.
}

\begin{document}
\maketitle
\flushbottom

\section{Introduction}
There are various situations in physics where one has to consider a manifold endowed with two different vielbein fields. Obviously, this appears to be the case in bimetric theories, theories where two different metrics are defined on the same space-time manifold \cite{Bimet}. Each of these metrics can then be described by a different vielbein. This is also true even if one of the two metrics is not dynamical. It also applies to non linear massive gravity (for recent reviews see \cite{Rubakov:2008nh,Hinterbichler:2011tt}), which is nothing else than a special class of bigravity, and in particular it applies to the recently introduced massive gravity theories of de Rham-Gabadadze-Tolley (dRGT in the following) \cite{deRham:2010kj,deRham:2010ik,deRham:2011rn} as well as to the extension of these to the dynamical bimetric case \cite{Hassan:2011tf,Hassan:2011zd}. A similar situation also occurs when one expands General Relativity around a fixed background metric and expresses both the background and the dynamical metric in terms of vielbeins. This is the starting point of many works dealing with quantum gravity (see e.g. \cite{Deser:1974cy,Woodard:1984sj}).

Considering such situations, let us define, in arbitrary $D$ dimensions, $\gv^{A}$ and $\fv^{A}$ to be two  bases of 1-forms obeying at every space-time point \footnote{Our convention is that Greek letters denote space-time indices, while capital Latin letters denote Lorentz indices that are moved up and down with the canonical Minkowski metric $\eta_{AB}$}
\be
g^{\mu\nu}\gv^{A}_{\hphantom{A}\mu}\gv^{B}_{\hphantom{B}\nu}=f^{\mu\nu}\fv^{A}_{\hphantom{A}\mu} \fv^{B}_{\hphantom{B}\nu}=\eta^{AB}\ ,
\ee
or equivalently
\ba
&\eta_{AB}\gv^{A}{}_{\mu}\gv^{B}{}_{\nu}=g_{\mu\nu}\ , \label{def2}\\
&\eta_{AB}\fv^{A}{}_{\mu}\fv^{B}{}_{\nu}=f_{\mu\nu}\ , \label{def3}
\ea
where $g_{\mu \nu}$ and $f_{\mu \nu}$ are respectively the metrics associated with the vielbeins. 
 We will also need the vectors $e_{A}$ and $\fmoinsunvierbein_{A}$, respectively dual to the 1-forms $\gv^{A}$ and $\fv^A$, that verify
\ba \gv^{A}(e_{B})=\gv^{A}{}_{\mu}e_{B}{}^{\mu}=\delta^{A}{}_{B}\ , \\
\fv^{A}(\fmoinsunvierbein_{B})=\fv^{A}{}_{\mu}\fmoinsunvierbein_{B}{}^{\mu}=\delta^{A}{}_{B}\ .
\ea
For future use, let us rewrite the above relations (and consequences thereof) using matrix notations. We have
\ba
f&=& L^t \eta L \label{fL}\ ,\\
f^{-1} &=& l^t \eta l \label{fmL}\ ,\\
\mathbb{1}_D&=&l^t L = L l^t = L^t l = l L^t \label{lL}\ , 
\ea
where $\mathbb{1}_D$ is the $D \times D$ identity matrix, $m^t$ denotes the matrix transpose of the matrix $m$, $\eta$ is just $diag(-1,1,\cdots,1)$ and the same relations hold between $\gv$, $e$ and $g$ respectively.

The defining relations (\ref{def2}) and (\ref{def3}) imply the gauge symmetry 
\ba \label{rot1}
\gv^{A}{}_\mu &\rightarrow& \Lambda^A_{\hphantom{A}C} \gv^{C}{}_\mu\ , \label{rot1} \\
\fv^{B}{}_\mu &\rightarrow& \tilde{\Lambda}^B_{\hphantom{B}D} \fv^{D}{}_\mu\ , \label{rot2}
\ea
with $\Lambda^A_{\hphantom{A}C}$ and $\tilde{\Lambda}^B_{\hphantom{B}D}$ Lorentz matrices. 

It is often convenient to ask for a ``symmetry" condition on the vielbeins which reads 
\ba \label{CONS} 
e^{\hphantom{A}\mu}_A \fv_{B \mu} = e^{\hphantom{B}\mu}_B \fv_{A \mu}\ .
\ea

Notice that this condition can also be written as $g^{\mu\nu}E^{A}{}_{\mu}L^{B}{}_{\nu}=g^{\mu\nu}E^{B}{}_{\mu}L^{A}{}_{\nu}$ and that Ref. \cite{Hoek:1982za} uses an equivalent form which reads $E^{A}{}_{\mu}L_{A\nu}=E^{A}{}_{\nu}L_{A\mu}$.   

In the recent discussions about massive gravity, such a condition has been used to ensure the existence of, and express, the matrix square root of $g^{-1}f$ which enters in a crucial way in the definition of dRGT theory (see e.g. \cite{Chamseddine:2011mu,Volkov:2012wp}).
Indeed, whenever condition (\ref{CONS}) holds, $\gamma$ defined as  
\ba \label{gammavier}
\gamma^\mu_{\hphantom{\mu} \nu} = e^{\hphantom{A}\mu}_A \fv^A_{\hphantom{A}\nu}
\ea
verifies the defining equation of the matrix square root of $g^{-1}f$ given by 
\ba \label{gammadef} 
\gamma^\mu_{\hphantom{\mu} \sigma}\gamma^\sigma_{\hphantom{\sigma} \nu} = g^{\mu \sigma} f_{\sigma \nu}\ .
\ea
It has also been argued by Hinterbichler and Rosen \cite{Hinterbichler:2012cn} that, in the vielbein reformulation of dRGT theories, condition (\ref{CONS}) is obtained as a consequence of field equations. To prove this, they use a decomposition of an arbitrary matrix $M$ (representing some unconstrained arbitrary vielbein  multiplied by $\eta$)
  as 
 \ba \label{Mlambdas} \label{POLARLOREN}
 M = \lambda s\ ,
 \ea
where $\lambda$ is  a Lorentz matrix and $s$ is a symmetric matrix. This is reminiscent of the so-called polar decomposition stating that an arbitrary invertible matrix can be written as the product of an orthogonal matrix with a symmetric matrix. However we will show that such a decomposition does not hold in general if one replaces the orthogonal matrix by  a Lorentz transformation. This makes in particular the argument of Ref. \cite{Hinterbichler:2012cn} incomplete.

Furthermore, in massive gravity as well as in perturbative quantum gravity condition (\ref{CONS}) has been used as a gauge condition.
 In the quantum gravity context, this gauge (sometimes dubbed Deser-van Nieuwenhuizen gauge in reference to \cite{Deser:1974cy}) has  been first introduced via a gauge fixing term in the action and dealt with perturbatively \cite{Deser:1974cy,Woodard:1984sj}. It was then later argued that this gauge can be set ``non perturbatively", i.e. that given a set of arbitrary vielbeins 
 $\gv^{A}$ and $\fv^{A}$ that do not fullfill condition (\ref{CONS}), one can always Lorentz rotate them as in (\ref{rot1}), and (\ref{rot2}) to define a new set of vielbeins obeying this condition \cite{Hoek:1982za} (with the consequence that the corresponding gauge would not suffer from Gribov-like ambiguities). Interestingly enough, the same statements  have also been made in the context of massive gravity. Indeed, there as well the condition (\ref{CONS}) has been used ``perturbatively" (i.e. in the case when both metrics $g_{\mu \nu}$ and $f_{\mu \nu}$ are close to one another, see e.g. \cite{Chamseddine:2011mu}), but it has also been argued that condition 
 (\ref{CONS}) can be reached as a (Lorentz) gauge choice for arbitrary metrics \cite{Volkov:2012wp}.
  This contradicts various other statements made in the literature, for example in Ref. \cite{Woodard:1984sj}, where it is stated that gauge (\ref{CONS}) cannot be set beyond perturbation theory. Settling this contradiction, as we intend to do here, will also illuminate issues discussed in the previous paragraph, since, as we will show, to set (\ref{CONS}) via suitable Lorentz rotations of the vielbeins involves a decomposition similar to (\ref{Mlambdas}). 

To be precise, the purpose of this note is to determine when and how the condition (\ref{CONS}) can be enforced, as well as when  the  decomposition (\ref{Mlambdas}) holds. These questions, beyond their mathematical interest, are especially important for massive gravity. Indeed, one can argue that the vielbein formulation of dRGT theories has several advantages over their metric formulations. First of all, it allows a simple extraction of what plays the role of the Hamiltonian constraint \cite{Hinterbichler:2012cn}. Second, in some cases it also allows to dynamically derive the existence of the square root of $g^{-1}f$ that has to be assumed or enforced by Lagrange multipliers in the metric formulation \cite{Hinterbichler:2012cn,us}. Finally, the frame formulation permits a simple discussion of the constraints and the counting of dynamical degrees of freedom in the Lagrangian framework \cite{us}. In this formulation, relation (\ref{CONS}) plays a key role, and it is important to know whether it can be obtained by Lorentz gauge transformations, or it needs additional constraints to be imposed.

This paper is organized as follows. In the next section, we will discuss necessary and sufficient conditions for  (\ref{CONS}) and (\ref{Mlambdas}) to hold.  Then, in section 3, using results on matrix square roots, we will spell out sufficient conditions to achieve  (\ref{CONS}) and (\ref{Mlambdas}). In the next sections we will discuss the specific cases of $D=2$, $D=3$, and $D=4$ space-time dimensions, and in particular some examples clarifying the results of section 3 as well as some left over cases. Finally we will quickly look at the stability of these conditions with respect to the dynamics of the system, i.e. we will discuss whether they are preserved under time evolution in some particular theories, and we will point out some consequences for massive gravity.

Before proceeding, let us mention a special choice for one of the metrics (say $f_{\mu \nu}$) and the associated vielbein $\fv^{A}$. This choice is made in some contexts (e.g. dRGT theories, but also perturbative quantum gravity). It amounts to first assuming that the  metric $f_{\mu \nu}$ is flat and takes the canonical form $\eta_{\mu \nu}$, i.e. 
\ba \label{GAUG1}
f_{\mu \nu} = \eta_{\mu \nu}\ ,
\ea
and then choosing $\fv^A = dx^A$, i.e. such that (in components)
\ba \label{GAUG2}
\fv^{A}_{\hphantom{A}\mu} = \delta^{A}_{\hphantom{A}\mu}\ .
\ea 
When the choice (\ref{GAUG1})-(\ref{GAUG2}) is made,
the constraint (\ref{CONS}) simply reads (labelling here  space-time indices and Lorentz indices with the same set of letters)
\ba \label{symeab}
e^{AB} = e^{BA}\ ,
\ea
stating that the vielbein $e^{A \mu}$ can be represented as a symmetric matrix. This choice will not be used to derive the results of this paper, but will just sometimes be considered as an example.

\section{Necessary and sufficient conditions}

Let us first try to set the constraint (\ref{CONS})  by using the freedom to Lorentz rotate independently the two sets of vielbeins $\fv^A$ and $e_{A}$. Considering two arbitrarily chosen vielbeins $e_{A}$ and $\fv^B$, assume that there  exist two Lorentz transformations $\Lambda^A_{\hphantom{A}B}$ and $\tilde{\Lambda}^A_{\hphantom{A}B}$ such that the matrix $S^{AB}$ defined by 
\ba \label {defS}
S^{AB} &=& \Lambda^A_{\hphantom{A}C} e^{C \mu} \fv^{D}{}_\mu \tilde{\Lambda}^B_{\hphantom{B}D} 
\ea
is symmetric. Defining $M$  as the matrix of components $M^{AB}$ given by \footnote{With our notations, Ref. \cite{Hoek:1982za} uses rather $\ell^{A\mu}E^{B}{}_{\mu}$ as a starting point.} 
\ba \label{defM}
M^{AB} = e^{A \mu} \fv^{B}{}_\mu
\ea (note that this definition implies that $M$ is invertible), the above equality (\ref{defS}) reads in matricial notations 
\ba
S &=& \Lambda M \tilde{\Lambda}^t\ .
\ea
 Multiplying  it on the right by $ \left(\tilde{\Lambda}^t\right)^{-1}$ and on the left by $\tilde{\Lambda}^{-1}$ we get 
 \ba  \tilde{\Lambda}^{-1} S \left(\tilde{\Lambda}^{t}\right)^{-1} &=& \tilde{\Lambda}^{-1} \Lambda M\ .
\ea
For $S$ to be symmetric, the matrix on the left hand side above should be symmetric, call it $s$. Defining the Lorentz transformation $\lambda $ by $\lambda = \Lambda^{-1} \tilde{\Lambda}$ we get that the invertible matrix $M$ should be written as in Eq.(\ref{Mlambdas}). Being a Lorentz transformation, $\lambda$  verifies 
\ba \label{rellam}
\lambda^t \eta \lambda = \eta = \lambda \eta \lambda^t\ .
\ea
 As we already stated, a decomposition
 such as in Eq.(\ref{Mlambdas}) does not hold in general (in constrast to the polar decomposition). Indeed, rewriting (\ref{POLARLOREN}) as $\lambda =  M s^{-1}$ and inserting this into (\ref{rellam}) we get after some trivial manipulation, that $M$ and $s$ should fullfill the necessary condition
\ba
\left(\eta s\right)\left( \eta s \right) = \eta  M^t  \eta M\ .
\ea
Running backward the above argument it is easy to see that the above condition is also sufficient (just because the matrix defined as 
$M s^{-1}$ will be a Lorentz  transformation). Hence we have proven the following proposition.

\begin{prop} \label{prop1}
 An arbitrary invertible matrix $M$ can be decomposed as $M = \lambda s$, $\lambda$ being the matrix of a Lorentz transformation and $s$ a symmetric matrix, if and only if (i) the real matrix $\eta M^{t} \eta M $ has a real square root, and (ii) at least one such square root can be written as the product of $\eta$ with a symmetric matrix.  

\end{prop}

In particular, when $M$ is given by (\ref{defM}), we have (using relations (\ref{fL})-(\ref{lL}) as well as  definition (\ref{defM}))
\ba \label{relM}
\eta  M^t  \eta M = \eta L g^{-1} f l^t \eta = \fmoinsunvierbein f g^{-1}\fmoinsunvierbein^{-1}= \left(\fv g^{-1} f \fv^{-1} \right)^t.
\ea
 So if $g^{-1} f$ has a square root $\gamma$, then (i) above holds: a square root of $\eta  M^t \eta M$ being then given by $\left(\fv \gamma \fv^{-1}\right)^t$. We then prove the following proposition, 
\begin{prop}
Given two metrics $g_{\mu \nu}$ and $f_{\mu \nu}$, $g^{-1} f$ has a square root $\gamma$ such that $\gamma = f^{-1} s$, with s a symmetric matrix, if and only if the matrix $M$ defined by (\ref{defM}) (and which verifies relations (\ref{relM}))  is such that the real matrix $\eta M^{t} \eta M $ has a real square root which can be written as the product of $\eta$ by a symmetric matrix.
\end{prop}

\begin{proof}
We first assume that $g^{-1}f$ can be written as $g^{-1} f = \left(f^{-1} s\right)^2$ with $s$ a symmetric matrix. 
Then using this hypothesis into the first equality of (\ref{relM}) we get 
\ba
\eta (M^t) \eta M &=& \eta \fv f^{-1} s f^{-1} s \fvm^t \eta \nonumber \\
&=& \fvm s f^{-1} s \fvm^{t} \eta \nonumber \nonumber \nonumber \\
&=& \fvm s \fvm^t \fv  f^{-1} s \fvm ^t \eta \nonumber \nonumber \\
&=& \fvm s \fvm^t \eta \fvm s \fvm ^t \eta \nonumber \nonumber \\
&=& \eta \left( \eta \fvm s \fvm^t \eta \right) \eta \left( \eta \fvm s \fvm^t \eta \right). \nonumber
\ea 
The matrix $\eta \fvm s \fvm^t \eta$ being symmetric, this proves one side  of the equivalence. 
Conversely, we assume that there exists a symmetric matrix $s'$ such that 
$\eta \left( M^t \right) \eta M =\left( \eta s'\right)^2$. Then $g^{-1} f$ is given by 
\ba
g^{-1} f &=& \ell^{t}\eta(\eta M^t\eta M) \eta L \nonumber \\
&=& \fvm^t  s' \eta s' \eta L \nonumber \\
&=& \fvm^t  s' \eta L \fvm^t s' \fvm f \nonumber \\
&=& \left (\fvm^t  s' \fvm f\right) \left( \fvm^t   s' \fvm f \right) \nonumber \\
&=& \left (f^{-1} f \fvm^t  s' \fvm f\right) \left( f^{-1} f \fvm^t   s' \fvm f \right). \nonumber
\ea
Noticing that the matrix $f \fvm^t  s' \fvm f$ is symmetric ends the proof.
\end{proof}
Hence, gathering the above results, we have proven the following statement.
\begin{prop} \label{prop2}
There exist vielbeins $e_{A}{}^{\mu}$ and $\fv^{B}{}_{\nu}$ corresponding to the metrics $g_{\mu\nu}$ and $f_{\mu\nu}$ respectively (i.e. $\eta^{AB}e_{A}{}^{\mu}e_{B}{}^{\nu}=g^{\mu\nu}$ and $\eta_{AB}\fv^{A}{}_{\mu}\fv^{B}{}_{\nu}=f_{\mu\nu}$) such that $e_{A}{}^{\mu}\fv_{B\mu}=e_{B}{}^{\mu}\fv_{A\mu}$, if and only if there exists  a real matrix $\gamma$ 
such that (i)  $\gamma^{\mu}{}_{\rho}\gamma^{\rho}{}_{\nu}=g^{\mu\rho}f_{\rho\nu}$ (i.e. $\gamma^{2}=g^{-1}f$), and (ii) $f\gamma$ symmetric. 
\end{prop}
\begin{proof}[Direct proof]
Suppose first we have vielbeins $e_{A}$ and $L^{B}$ satisfying the above symmetry property. Then 
\be
\begin{split}
g^{\mu\rho}f_{\rho\nu}&=\eta^{AB}e_{A}{}^{\mu}e_{B}{}^{\rho}\eta_{CD}L^{C}{}_{\rho}L^{D}{}_{\nu}\\
&=\eta^{AB}e_{A}{}^{\mu}L^{D}{}_{\nu}e_{B}{}^{\rho}L_{D\rho}\\
&=\eta^{AB}e_{A}{}^{\mu}L^{D}{}_{\nu}e_{D}{}^{\rho}L_{B\rho}
=e_{A}{}^{\mu}L^{A}{}_{\rho}e_{D}{}^{\rho}L^{D}{}_{\nu}\ ,
\end{split}
\ee
and if we define $\gamma^{\mu}{}_{\nu}\equiv e_{A}{}^{\mu}L^{A}{}_{\nu}\in \mathbb{R}$ we get $g^{\mu\rho}f_{\rho\nu}=\gamma^{\mu}{}_{\rho}\gamma^{\rho}{}_{\nu}$. Moreover
\be
\begin{split}
f_{\mu\rho}\gamma^{\rho}{}_{\nu}&= \eta_{AB}L^{A}{}_{\mu}L^{B}{}_{\rho}e_{C}{}^{\rho}L^{C}{}_{\nu}\\
&=L^{A}{}_{\mu}e_{C}{}^{\rho}L_{A\rho}L^{C}{}_{\nu}\\
&=L^{A}{}_{\mu}e_{A}{}^{\rho}L_{C\rho}L^{C}{}_{\nu}\\
&=\eta_{BC}L^{B}{}_{\nu}L^{C}{}_{\rho}e_{A}{}^{\rho}L^{A}{}_{\mu}=f_{\nu\rho}\gamma^{\rho}{}_{\mu}\ ,
\end{split}
\ee
which shows that the matrix $f\gamma$ is symmetric. Notice that this is equivalent to $\gamma f^{-1}$ symmetric.  Conversely, suppose we have a real matrix $\gamma$ such that $\gamma^{2}=g^{-1}f$ and $f\gamma$ symmetric. We start by choosing an arbitrary vielbein $L^{A}$ for the metric $f_{\mu\nu}$ i.e. $f_{\mu\nu}=\eta_{AB}L^{A}{}_{\mu}L^{B}{}_{\nu}$, and we denote by $\ell_{B}$ its dual vector i.e. $f^{\mu\nu}=\eta^{AB}\ell_{A}{}^{\mu}\ell_{B}{}^{\nu}$. We then define $e_{A}{}^{\mu}\equiv\gamma^{\mu}{}_{\nu}\ell_{A}{}^{\nu}$. This implies that
\be
\begin{split}
\eta^{AB}e_{A}{}^{\mu}e_{B}{}^{\nu}&=\eta^{AB}\gamma^{\mu}{}_{\rho}\ell_{A}{}^{\rho}\gamma^{\nu}{}_{\sigma}\ell_{B}{}^{\sigma}\\
&=f^{\rho\sigma}\gamma^{\mu}{}_{\rho}\gamma^{\nu}{}_{\sigma}\\
&=(\gamma f^{-1}\gamma^{t})^{\mu\nu}\ .
\end{split}
\ee
But the symmetry of $\gamma f^{-1}$ implies that $(\gamma f^{-1})^{t}=f^{-1}\gamma^{t}=\gamma f^{-1}$ so 
\be
\eta^{AB}e_{A}{}^{\mu}e_{B}{}^{\nu}=(\gamma^{2}f^{-1})^{\mu\nu}=g^{\mu\nu}\ ,
\ee
and $e_{A}$ is a well-defined vielbein for the metric $g_{\mu\nu}$. Notice that this definition tells us $\gamma^{\mu}{}_{\nu}=e_{A}{}^{\mu}L^{A}{}_{\nu}$. It remains to be shown that these vielbeins have the required symmetry property. We start from the symmetry of $f\gamma$
\be
f_{\mu\rho}\gamma^{\rho}{}_{\nu}=f_{\nu\rho}\gamma^{\rho}{}_{\mu}\ ,
\ee
which we can rewrite 
\be
\eta_{AB}L^{A}{}_{\mu}L^{B}{}_{\rho}e_{C}{}^{\rho}L^{C}{}_{\nu}=\eta_{AB}L^{A}{}_{\nu}L^{B}{}_{\rho}e_{C}{}^{\rho}L^{C}{}_{\mu}\ .
\ee
Multiplying by $\ell_{D}{}^{\mu}\ell_{E}{}^{\nu}$ we get $e_{E}{}^{\rho}L_{D\rho}=e_{D}{}^{\rho}L_{E\rho}$ and this completes the proof.
\end{proof}

As we just showed the hypotheses  (i) of Propositions \ref{prop1} and \ref{prop2} are that a certain real (invertible) matrix has a real square root.  It is however well known that not all real invertible matrices have real square roots (see e.g. \cite{Higham,Gallier}) and we will later recall what are the necessary and sufficient conditions for this to occur. In our case, though, the matrix which should have a square root is not totally arbitrary. For example, in Proposition \ref{prop1} it must be of the form $\eta \left(M^{t}\right) \eta M $. This alone does however not ensure the existence of a square root. For example, choosing 
\be
M=\left(
\begin{array}{cccc}
 0 & -1 & 0 & 0 \\
 -3 & 0 & 0 & 0 \\
 0 & 0 & 2 & 0 \\
 0 & 0 & 0 & 1 \\
\end{array}
\right),
\ee
we get 
\be
\eta M^{t}\eta M =\left(
\begin{array}{cccc}
 -9 & 0 & 0 & 0 \\
 0 & -1 & 0 & 0 \\
 0 & 0 & 4 & 0 \\
 0 & 0 & 0 & 1\\
\end{array}
\right)
\ee
which doesn't have any real square roots. Indeed, such a $4 \times 4$ diagonal matrix with four distinct eigenvalues has $2^4$ square roots which are given here by $diag\left(\pm 3i, \pm i, \pm 2 ,\pm 1 \right)$. None of them is real. Hence the decomposition (\ref{POLARLOREN}) can at best hold for a restricted set of matrices.  

We thus see that considering the matrix $M$ above as given by the form (\ref{defM}) invalidates the result of Ref. \cite{Hoek:1982za}. Notice that, if one makes now the simple choice (\ref{GAUG1})-(\ref{GAUG2}) (and considering equation (\ref{relM})), our example involves a ``mismatch" between the time directions of the two metrics $f_{\mu \nu}$ and $g_{\mu \nu}$. However, beyond perturbation theory there is no reason to think that these time directions should coincide or even be compatible.  We will come back to this question later. Notice further that perturbatively, if $g = f+h$, with $h$ small, then to the first order in $h$ $g^{-1} f = (\mathbb{1}_{D}-1/2 f^{-1} h)^2=(\mathbb{1}_{D}-1/2 g^{-1} h)^2$, and so the assumptions (i) and (ii) of Proposition \ref{prop2} are always true perturbatively.

\label{section1}

\section{Sufficient conditions}
Here, in order to formulate simple sufficient conditions  allowing to obtain  (\ref{CONS}) and (\ref{POLARLOREN}), we will discuss the precise relation between hypotheses  (i) and (ii) of Propositions \ref{prop1} and \ref{prop2}. We need to recall how square roots of real matrices are obtained.  We first use the following theorem (that we quote here from Ref.\cite{Higham}).

\begin{theo}\label{theo1}
Let $A$ be an invertible real square matrix (of arbitrary dimension). If $A$ has no real negative eigenvalues, then there are precisely $2^{r+c}$ real square roots of $A$ which are polynomial functions of $A$, where $r$ is the number of distinct eigenvalues of $A$ and $c$ is the number of distinct complex conjugate eigenvalue pairs. If $A$ has a real negative eigenvalue, then $A$ has no real square root which is a  polynomial function of $A$.
\end{theo} 
Let us first use this theorem to prove that (i) of Proposition \ref{prop1} (respectively Proposition \ref{prop2}) implies (ii) of the same proposition whenever the matrix $\eta M^t \eta M$ (respectively the matrix $g^{-1} f$) has no real negative eigenvalues.  To see this, just consider a real matrix $A$ with no negative eigenvalues, given by the product of two symmetric invertible matrices $S$ and $S'$. By virtue of the above theorem, we know that this matrix has at least one real square root which is a polynomial function of $A$, that we note $F(A)$. One then has
\ba \label{defF}
F(A) = \sum c_k A^k\ , 
\ea 
where the sum runs over a finite number of integers $k$, and $c_k$ are real numbers.
Using the fact that  $A = SS'$, one then has 
\ba
F(A) = S\left(c_0 S^{-1} + \sum_{k \geq 1} c_k \left[S'SS' \cdots SS'\right]_k\right)
\ea 
where the term $\left[S'SS' \cdots SS'\right]_k $contains $k$ factors of $S'$ and $k-1$ factors of $S$, and is a symmetric matrix. 
This means that that the square root $F(A)$ is given by the product of $S$ by a symmetric matrix. It is enough to prove our assertion by choosing $S$ to be given by $\eta$ and $S'$ to be given by $M^t \eta M$ (respectively $S$ given by $f^{-1}$ and $S'$ to be given by $f g^{-1} f$). Hence, using the above result, and Propositions \ref{prop1} and \ref{prop2} we have shown the following two propositions
\begin{prop}
\label{prop4}
A sufficient condition for an arbitrary invertible real matrix $M$ to  be decomposed as $M = \lambda s$, $\lambda$ being the matrix of a Lorentz transformation and $s$ a symmetric matrix, is that the  matrix $\eta \left(M^{t}\right) \eta M $ has no negative eigenvalues.
\end{prop}

\begin{prop} \label{prop5}
A sufficient condition for the existence of vielbeins $e_{A}{}^{\mu}$ and $\fv^{B}{}_{\nu}$ corresponding to the metrics $g_{\mu\nu}$ and $f_{\mu\nu}$ respectively (i.e. $\eta^{AB}e_{A}{}^{\mu}e_{B}{}^{\nu}=g^{\mu\nu}$ and $\eta_{AB}\fv^{A}{}_{\mu}\fv^{B}{}_{\nu}=f_{\mu\nu}$) such that $e_{A}{}^{\mu}\fv_{B\mu}=e_{B}{}^{\mu}\fv_{A\mu}$, is that the matrix $g^{-1} f$ has no negative eigenvalues.
\end{prop}

If $A$ has one (or more) real negative eigenvalue, Theorem \ref{theo1} does not imply that $A$ does not have a real square root, but just that such a square root cannot be a polynomial function of $A$. In order to enunciate the necessary and sufficient conditions for a real matrix to have a real square root, one first needs to introduce the so-called Jordan decomposition of a matrix. It uses Jordan blocks which can be defined as $r \times r$ matrices wich are of the form $J_{(r,z)}$ given by (for $r\geq 2$)
\ba
J_{(r,z)} =  \left(\begin{array}{ccccc}
 z & 1  &0& \cdots & 0 \\
 0 & z & 1 &  \ddots  & \vdots \\
 \vdots & \ddots & \ddots & \ddots & 0 \\
 \vdots &  \ddots & \ddots & z & 1 \\
 0 & \cdots  &  \cdots & 0 & z \\
\end{array}\right)
\ea
where $z$ is a complex number, and one has $J_{(1,z)}=\left(z\right)$ for $r=1$. One can then show that for an arbitrary $n \times n$ matrix $A$, there exists an invertible matrix $P$ (possibly complex), and a matrix $J$ such that 
\ba \label{Jordan1}
PA P^{-1}=J 
\ea
and the matrix J is a so called Jordan matrix of the form 
\ba \label{Jordan2}
J = diag \left(J_{(r_1,z_1)},J_{(r_2,z_2)}, \cdots, J_{(r_k,z_k)}\right),
\ea
where $k$ is an integer and the matrices $J_{(r_j,z_j)}$
 are called the Jordan blocks of $J$. For a given matrix $A$, the number of Jordan blocks, the nature of the distinct Jordan blocks, and the number of times a given Jordan block occurs in the Jordan matrix $J$ are uniquely determined. Moreover, the $z_i$ are the eigenvalues of $A$. 
 One can further show that a given Jordan block $J_{(r,z)}$ with $z \neq 0$, has precisely two upper triangular square roots, $j^{\pm}_{(r,z)}$, which are in addition polynomial functions of $J_{(r,z)}$ \cite{Higham}. These can be used to find {\it all} the square roots (possibly complex) of a given matrix
 using the following theorem.
 \begin{theo} \label{theo3}
 Let $A$ be a $n \times n$ complex matrix which has a Jordan decomposition given by (\ref{Jordan1})-(\ref{Jordan2}), then all the square roots (which may include complex matrices) of $A$ are given by the matrices $P^{-1} U^{-1} diag \left(j^\pm_{(r_1,z_1)},j^\pm_{(r_2,z_2)}, \cdots, j^\pm_{(r_k,z_k)}\right) U P$, where $U$ is an arbitrary matrix which commutes with $J$.
\end{theo} 
 The Jordan blocks of a matrix also play a crucial role in the following theorem  which gives the necessary and sufficient condition for a real matrix to have a real square root (see e.g. \cite{Gallier}).
 \begin{theo} \label{theo2}
Let $A$ be an invertible real square matrix (of arbitrary dimension). The matrix $A$ has a real square root if and only if for each of its negative eigenvalues $z_i$, the number of identical Jordan block $J_{(r_i,z_i)}$ where this eigenvalue occurs in the Jordan decomposition of the matrix $A$ is even.   
\end{theo}

In the following, we will use the above theorems to discuss in detail the cases\footnote{Note however that according to Theorem \ref{theo2} these cases  should be of zero measure with respect to those which are included.} which are not covered by our Propositions \ref{prop4} and \ref{prop5}. Namely, we will ask if it possible for a matrix to fullfill condition (i) (of Propositions \ref{prop1} and \ref{prop2}) without obeying condition (ii) (of the same propositions). We will do it for various space-time dimensions, starting with the two dimensional case, which has less interest as far as gravity is concerned, but where results useful for the other cases can be derived. In this case we will also be able to give an explicit proof of the propositions of section \ref{section1}. 

\label{newsect2}

\section{Two dimensional case}
A certain number of the results derived before can easily be obtained in two dimensions by an explicit calculation. Consider first the decomposition (\ref{POLARLOREN}). We ask if an arbitrary $2 \times 2$ invertible matrix $M$ given by 
\ba
M= \left(
\begin{array}{cc}
 A & B \\
 C & D\\
\end{array}
\right)
\ea
can be written as (beginning here with proper orthochronous Lorentz transformations)
\ba\label{dec22}\left(
\begin{array}{cc}
 A & B \\
 C & D\\
\end{array}
\right) = 
\left(
\begin{array}{cc}
 c & s \\
 s  & c \\
\end{array}
\right)\left(
\begin{array}{cc}
 a & b \\
 b & d\\
\end{array}
\right)
\ea
where $c = \cosh \psi$ and $s= \sinh \psi$ (and $\psi$ a real number). Expanding the  matrix product in the right hand side, we obtain a system of 4 linear equations obeyed by the three coefficients $\{a, b, d \}$ which we can use, eliminating $b$, to get the necessary condition $(A-D) s =  (C-B) c$, which cannot hold for $|C-B| > |A-D|$. This obviously shows that the decomposition (\ref{dec22}) is not always possible\footnote{This conclusion can be extended easily with the same derivation to the case of a Lorentz transformation which is not proper and/or orthochronous.}, as we showed in a more general way in Proposition \ref{prop1}.

In two dimensions, one can also explicitly show that the condition (i) of Proposition \ref{prop1} always implies the condition (ii) of the same proposition. 
Indeed, consider a $2 \times 2$ matrix $m$, that is written as $m = \eta s$, with $s$ symmetric.  
Let us then assume that this matrix has a square root. According to the proof of Proposition \ref{prop4}, we know that if this matrix has no negative eigenvalues,
it has a square root which is a product of $\eta$ times a symmetric matrix.  Let us study the case where it has at least one  negative eigenvalue. In this case, according to Theorem \ref{theo2}, it must be of the form $m= P\ diag\left(-u,-u\right) P^{-1} = -u \mathbb{1}_2$, where $u$ is a positive non zero number \footnote{This means that $m$ has two identical one dimensional Jordan Block $(-u)$.} (note that such a matrix is indeed in the form $\eta s$).
It remains then to study all the square roots of 
\ba
m= \left(
\begin{array}{cc}
 -u & 0 \\
 0 & -u \\
\end{array}
\right)\ .
\ea
The matrix equation $\gamma^2 = m$ is easy to solve explicitly. We obtain that a real square root $\gamma$ is given by any of the matrices  
\ba
\gamma = \left(
\begin{array}{cc}
 \alpha & \beta \\
 - \frac{u + \alpha^2}{\beta} & -\alpha \\
\end{array}
\right) = \left(
\begin{array}{cc}
 -1 & 0 \\
 0 & 1 \\
\end{array}
\right)
\left(
\begin{array}{cc}
 -\alpha & -\beta \\
 - \frac{u + \alpha^2}{\beta} & -\alpha \\
\end{array}
\right) \label{squa2}
\ea
where $\beta$ and $\alpha$ are real numbers and $\beta$ is non zero. Choosing then $\alpha$ and $\beta$ which obey the constraint $u = \beta^2 - \alpha^2$ we find an infinite family of real matrix square roots of $m$  which are written in the form of the product of $\eta$ by a symmetric matrix. 
A  similar straightforward calculation can be made to prove that hypothesis (i) of Proposition \ref{prop2} implies (ii) of the same proposition.  In fact, it is easy to see that for every symmetric matrix 
\be
\left(
\begin{array}{cc}
 a & b \\
 b & c \\
\end{array}
\right)
\ee
with $ac-b^{2}<0$ there exist real $\alpha$, $\beta$ such that
\be
\left(
\begin{array}{cc}
 a & b \\
 b & c \\
\end{array}
\right)
 \left(
\begin{array}{cc}
 \alpha & \beta \\
 - \frac{u + \alpha^2}{\beta} & -\alpha \\
\end{array}
\right)
\ee
is symmetric i.e. such that $a\beta^{2}-2\alpha\beta b +c u+c\alpha^{2}=0$. Indeed, either $c\neq 0$ and the discriminant of the above second order polynomial equation with respect to $\alpha$, $\Delta_{\alpha}=4\beta^{2}(b^{2}-ac)-4c^{2}u$, is positive for large enough $\beta$, or $c=0$ in which case $b$ must be non-zero and $\alpha=\frac{a\beta}{2b}$ is an obvious solution.
This shows that in 2 dimensions, being able to choose zweibeins obeying (\ref{CONS}) is equivalent to the existence of a real square root of $g^{-1} f$. 

\section{Three dimensional case}
The results obtained in the previous section can be extended to the case of a spacetime with 3 dimensions, which has some relevance for physics and 
in particular massive gravity \cite{Bergshoeff:2009hq,Bergshoeff:2009aq,deRham:2011ca}. In three dimensions, the only cases which are not covered 
by Propositions \ref{prop4} and \ref{prop5} are the cases of real invertible matrices $A$ which have the form 
\ba
A= P^{-1} \left(
\begin{array}{ccc}
 -u & 0 & 0\\
 0 & -u & 0 \\
 0 & 0 & v \\
\end{array}
\right) P
\label{D=3}
\ea
where $u$ and $v$ are non zero positive real numbers, and $P$ is an invertible matrix. Notice that because $A$, $u$ and $v$ are real, $P$ may also be assumed to be real. Before going any further, notice that one can find $3 \times 3$ matrices $A$, in the form $A=\eta s$ with $s$ symmetric, having real square roots, but such that none of these square roots is the product of $\eta$ by a symmetric matrix. Indeed consider $A$ to be given by 
\ba
A=\left(
\begin{array}{ccc}
 7 & -4 & 4 \\
 4 & -3 & 2 \\
 -4 & 2 & -3 \\
\end{array}
\right)
=\left(
\begin{array}{ccc}
 -1 & 1 & -2 \\
 0 & 2 & -1 \\
 2 & 0 & 1 \\
\end{array}
\right)
\left(
\begin{array}{ccc}
 -1 & 0 & 0\\
 0 & -1 & 0 \\
 0 & 0 & 3 \\
\end{array}
\right)
\left(
\begin{array}{ccc}
 -1 & 1 & -2 \\
 0 & 2 & -1 \\
 2 & 0 & 1 \\
\end{array}
\right)^{-1}.
\ea
This matrix has the form of a product of $\eta$ with a symmetric matrix, but none of its real square roots, given by 
\ba
\left(
\begin{array}{ccc}
 -1 & 1 & -2 \\
 0 & 2 & -1 \\
 2 & 0 & 1 \\
\end{array}
\right)
\left(
\begin{array}{ccc}
 \alpha & \beta & 0\\
 -\frac{1+ \alpha^2}{\beta} &-\alpha & 0 \\
 0 & 0 & \pm \sqrt{3} \\
\end{array}
\right)
\left(
\begin{array}{ccc}
 -1 & 1 & -2 \\
 0 & 2 & -1 \\
 2 & 0 & 1 \\
\end{array}
\right)^{-1},
\ea
(with $\alpha$ and $\beta$ real numbers, $\beta$ non vanishing) has the same form. However, this example does not apply to the cases of interest here because $\eta A$ does not have the correct signature: instead of being of signature $(-,+,+)$ as e.g. a matrix of the form $s=M^{t}\eta M$, it is negative definite. 

In contrast we are going to show that (i) of Proposition \ref{prop1} (respectively Proposition \ref{prop2}) implies (ii) of the same proposition whenever the matrix $\eta M^t \eta M$ (respectively the matrix $g^{-1} f$) is of the form (\ref{D=3}). In order to do that let us assume (for the same reason as in section \ref{newsect2}) that $A=P^{-1}J P=S S'$ with $S$ and $S'$ two symmetric matrices of $(-,+,+)$ signature. The fact that $S'$ is symmetric implies that $PSP^{t}$ commutes with $J$ and thus it must be of the form
\be
PSP^{t}=
\left(
\begin{array}{ccc}
  S_{2} & 0 \\
  0 & r \\
\end{array}
\right)\ ,
\label{d3sig1}
\ee
with $S_{2}$ a symmetric two by two matrix and $r$ a real number such that $r\det(S_{2})\neq 0$. Since $PSP^{t}$ is of $(-,+,+)$ signature, it is obvious that $S_{2}$ cannot be negative definite. From the fact that $S'$ has $(-,+,+)$ signature we can infer that $JPSP^{t}=(PS)S'(PS)^{t}$ also has the same signature. But
\be
JPSP^{t}=
\left(
\begin{array}{ccc}
 -uS_{2} & 0 \\
 0 & vr \\
\end{array}
\right)\\ ,
\ee
and thus $S_{2}$ cannot be positive definite either. We therefore necessarily conclude that $S_{2}$ must have $(-,+)$ signature and that $r>0$. This means that there exists a two by two invertible matrix $U_{2}$ such that 
\be
U_{2}S_{2}U_{2}^{t}=
\left(
\begin{array}{ccc}
 -1 & 0 \\
 0 & 1 \\
\end{array}
\right)\ .
\label{d3sig2}
\ee
Now let us define 
\be
U=
\left(
\begin{array}{ccc}
 U_{2} & 0 \\
  0 & 1 \\
\end{array}
\right)\ .
\label{d3sig3} 
\ee
This matrix clearly commutes with $J$ and if we further define
\be
\gamma=P^{-1}U^{-1}\left(
\begin{array}{ccc}
 0 & \sqrt{u} & 0 \\
 -\sqrt{u} & 0 & 0 \\
 0 & 0 & \sqrt{v} \\
\end{array}
\right)
UP\ ,
\ee
we can see that $\gamma^{2}=P^{-1}JP=A$ and thus $\gamma$ is a real square root of $A$.  Furthermore it is easy to see using (\ref{d3sig1}), (\ref{d3sig2}) and (\ref{d3sig3}) that 
\be
S^{-1}\gamma = P^{t}U^{t}
\left(
\begin{array}{ccc}
 0 & -\sqrt{u} & 0 \\
 -\sqrt{u} & 0 & 0 \\
 0 & 0 & \sqrt{v} \\
\end{array}
\right)
UP\ ,
\ee
is symmetric. This provides a constructive proof of our statement.

\section{Four dimensional case}
Considering here the case of $4 \times 4$ real matrices, and using Theorems \ref{theo3} and \ref{theo2}, we have that the only real invertible  matrices $A$ that have at least one negative real eigenvalue and also have at least one real square root must have one of the following Jordan forms  
\ba
A = P^{-1} J_k P
\label{D=4}
\ea 
where $J_k$ is one of the Jordan matrices
\ba
J_1 &=&  diag(-u,-u,-v,-v) \\
J_2 &=& diag (-u,-u,v,w) \\
J_3 &=& diag \left(\left(\begin{array}{cc} -u& 0 \\0 &-u \end{array}\right),\pm \left(\begin{array}{cc}  v+ i w & 0 \\0 & v-i w \end{array}\right)\right) \\
J_4 &=& diag \left(\left(\begin{array}{cc} -u& 0 \\0 &-u \end{array}\right),\left(\begin{array}{cc} v & 1 \\0 &v  \end{array}\right)\right) \\
J_5 &=& diag \left(\left(\begin{array}{cc} -u& 1 \\0  &-u \end{array}\right),\left(\begin{array}{cc} -u & 1 \\0 &-u  \end{array}\right)\right)
\ea 
where $u,v$ and $w$ are positive real numbers, $u$ and $w$ are always non zero, and $v$ can only vanish in the case of $J_{3}$. Because $A$ is real, the invertible matrix $P$ may be chosen to be real in the $J_{1}$, $J_{2}$, $J_{4}$ and $J_{5}$ cases. The case of $J_{3}$ is a bit more tricky, but we can also assume $P$ to be real as long as we replace the Jordan matrix $J_{3}$ by its {\it real}\footnote{This is a particular case of a result usually known as the {\it real} Jordan decomposition of a real matrix.} counterpart
\be
J_{3}'= diag \left(\left(\begin{array}{cc} -u& 0 \\0 &-u \end{array}\right),\pm \left(\begin{array}{cc}  v & w \\-w & v \end{array}\right)\right)\ .
\ee

We will show here that results similar to the ones obtained above in the $D=2$ and $D=3$ cases hold for $D=4$ whenever $A$ is of the form (\ref{D=4}) and  $A=P^{-1}J_{k}P=SS'$ with $S$ and $S'$ two symmetric matrices of Lorentzian signature. We will look in turn at the different cases for what concerns $J_{k}$. Consider first the case where the matrix $A=SS'$ is diagonalizable over $\mathbb{R}$. One can show that this is a sufficient (and in fact also necessary) condition to be able to diagonalize (in the sense of forms) in a common basis the matrices $S^{-1}$ and $S'$ corresponding to two symmetric bilinear forms \cite{Uhlig}\footnote{If one of the two bilinear forms had a euclidean signature, then it would have been possible to diagonalize matrices corresponding to both forms in the same basis without any further assumption.}. In this common basis, each of the diagonal matrices corresponding to $S^{-1}$ and $S'$ has only one negative eigenvalue, and hence there is no way that $A=SS'$ can be equal or similar (in the mathematical sense) to $J_1$, which has four negative eigenvalues. This excludes the $J_{1}$ case from the start. 

The discussion of the $J_{2}$ case proceeds along the same lines as in the $D=3$ case. The fact that $S'$ is symmetric implies that $PSP^{t}$ commutes with $J_{2}$ and thus it must be of the form
\be
PSP^{t}=
\left(
\begin{array}{ccc}
  S_{2} & 0 \\
  0 & S'_{2} \\
\end{array}
\right)\ ,
\label{d4sig1}
\ee
with $S_{2}$ and $S'_{2}$ symmetric two by two matrices such that $\det(S_{2})\det(S'_{2})\neq 0$. Notice that $S'_{2}$ must be diagonal whenever $v\neq w$. Since $PSP^{t}$ is of $(-,+,+,+)$ signature, it is obvious that $S_{2}$ and $S'_{2}$ cannot be negative definite. From the fact that $S'$ has $(-,+,+,+)$ signature we can infer that $J_{2}PSP^{t}=(PS)S'(PS)^{t}$ also has the same signature. But
\be
J_{2}PSP^{t}=
\left(
\begin{array}{ccc}
 -uS_{2} & 0 \\
 0 & diag(v,w)S'_{2} \\
\end{array}
\right)\\ ,
\ee
and thus $S_{2}$ cannot be positive definite either. We therefore necessarily get that $S_{2}$ must have $(-,+)$ signature and that $S'_{2}$ must be positive definite. In particular this means that there exist two by two invertible matrices $U_{2}$ and $V_{2}$ such that 
\be
U_{2}S_{2}U_{2}^{t}=
\left(
\begin{array}{ccc}
 -1 & 0 \\
 0 & 1 \\
\end{array}
\right)
\quad\text{and}\quad
V_{2}S'_{2}V_{2}^{t}=\mathbb{1}_{2}\ ,
\label{d4sig2}
\ee
and whenever $v \neq w$, we can further assume that $V_2$ is diagonal (this is because $S'_2$ is then diagonal and positive definite).
Now let us define 
\be
U=
\left(
\begin{array}{ccc}
 U_{2} & 0 \\
  0 & V_{2} \\
\end{array}
\right)\ .
\label{d4sig3}
\ee
This matrix clearly commutes with $J_{2}$ and if we further define
\be
\gamma=P^{-1}U^{-1}\left(
\begin{array}{cccc}
 0 & \sqrt{u} & 0 & 0 \\
 -\sqrt{u} & 0 & 0 & 0 \\
 0 & 0 & \sqrt{v} & 0 \\
 0 & 0 & 0 & \sqrt{w}\\
\end{array}
\right)
UP\ ,
\ee
we can see that $\gamma^{2}=P^{-1}J_{2}P=A$ and thus $\gamma$ is a real square root of $A$. Analogously to what has been done in the previous section, using (\ref{d4sig1}), (\ref{d4sig2}) and (\ref{d4sig3}), it is also easy to see that 
\be
S^{-1}\gamma = P^{t}U^{t}\left(
\begin{array}{cccc}
 0 & -\sqrt{u} & 0 & 0 \\
 -\sqrt{u} & 0 & 0 & 0 \\
 0 & 0 & \sqrt{v} & 0 \\
 0 & 0 & 0 & \sqrt{w}\\
\end{array}
\right) UP\ ,
\ee
is symmetric. This shows, as in the $D=3$ case, that whenever $A=P^{-1}J_{2}P$ and hypothesis (i) of Proposition \ref{prop1} (respectively Proposition \ref{prop2}) is verified, hypothesis (ii) of the same proposition is also verified. 

The three remaining cases ($J_3, J_4$ and $J_5$) actually never occur as long as we assume that $A$ is the product of two symmetric matrices of Lorentzian signature ($A=SS'$), as we now show. In the $J_{3}$ case, it is easier to work with the real Jordan form of $A$ i.e. $J'_{3}$. In order to understand the implications of the symmetry of $S'$ we need to introduce the matrix
\be
\sigma=
\left(
\begin{array}{cccc}
 1 & 0 & 0 & 0 \\
 0 & 1 & 0 & 0 \\
 0 & 0 & 0 & 1 \\
 0 & 0 & 1 & 0 \\
\end{array}
\right)\ .
\ee
Then it is easy to see that, given the particular form of $J'_{3}$, the symmetry of $S'$ implies that $PSP^{t}\sigma$ commutes with $J'_{3}$. Therefore
\be
PSP^{t}\sigma=
\left(
\begin{array}{ccc}
 S_{2} & 0 & 0 \\
  0 & r & r' \\
  0 & -r' & r \\
\end{array}
\right)\quad 
\text{or equivalently}\quad
PSP^{t}=
\left(
\begin{array}{ccc}
 S_{2} & 0 & 0 \\
  0 & r' & r \\
  0 & r & -r' \\
\end{array}
\right)\ ,
\ee
with $S_{2}$ a symmetric two by two matrix and $r$, $r'$ real numbers such that $\det(S_{2})(r^{2}+r'^{2})\neq 0$. Since the signature of $S$ is $(-,+,+,+)$ and $r^{2}+r'^{2}>0$ (which is the opposite of the determinant of the $2 \times 2$ lower block in the right matrix above), $S_{2}$ must be positive definite. But we also know that the signature of $J'_{3} PSP^{t}=(PS)S'(PS)^{t}$ is $(-,+,+,+)$ and since 
\be
J'_{3}PSP^{t}=
\left(
\begin{array}{ccc}
 -uS_{2} & 0 & 0 \\
  0 & * & * \\
  0 & * & * \\
\end{array}
\right)\ ,
\ee
$S_{2}$ cannot be positive definite and we have a contradiction. This proves by reductio ad absurdum that the $J_{3}$ case cannot occur in this context. A similar argument works for the $J_{4}$ case. Indeed the symmetry of $S'$ again implies that $PSP^{t}\sigma$ commutes with $J_{4}$. Therefore
\be
PSP^{t}\sigma=
\left(
\begin{array}{ccc}
 S_{2} & 0 & 0 \\
  0 & r & r' \\
  0 & 0 & r \\
\end{array}
\right)\quad 
\text{or equivalently}\quad
PSP^{t}=
\left(
\begin{array}{ccc}
 S_{2} & 0 & 0 \\
  0 & r' & r \\
  0 & r & 0 \\
\end{array}
\right)\ ,
\ee
with $S_{2}$ a symmetric two by two matrix and $r$, $r'$ real numbers such that $r^{2}\det(S_{2})\neq 0$. Since the signature of $S$ is $(-,+,+,+)$ and $r^{2}>0$, $S_{2}$ must be positive definite. But, with a similar argument as in the above case, we know that $S_{2}$ cannot be positive definite and we again stumble upon a contradiction. Finally the $J_{5}$ case can be handled in the same manner. Introducing
\be
\sigma'=
\left(
\begin{array}{cccc}
 0 & 1 & 0 & 0 \\
 1 & 0 & 0 & 0 \\
 0 & 0 & 0 & 1 \\
 0 & 0 & 1 & 0 \\
\end{array}
\right)\ ,
\ee
we can express the symmetry of $S'$ as the fact that $PSP^{t}\sigma'$ commutes with $J_{5}$. This in turn means that
\be
PSP^{t}\sigma'=
\left(
\begin{array}{cccc}
 a & b & c & d \\
 0 & a & 0 & c \\
 c & d & e & f \\
 0 & c & 0 & e \\
\end{array}
\right)\quad 
\text{or equivalently}\quad
PSP^{t}=
\left(
\begin{array}{cccc}
 b & a & d & c \\
 a & 0 & c & 0 \\
 d & c & f & e \\
 c & 0 & e & 0 \\
\end{array}
\right)\ ,
\ee
with $a$, $b$, $c$, $d$, $e$, $f$ real numbers such that  $ae-c^{2}\neq 0$. But $\det(PSP^{t})=(ae-c^{2})^{2}>0$ which is incompatible with the Lorentzian signature of $PSP^{t}$ and this excludes the last case. 

This lengthy discussion has shown that (i) of Proposition \ref{prop1} (respectively Proposition \ref{prop2}) implies (ii) of the same proposition whenever the matrix $\eta M^t \eta M$ (respectively the matrix $g^{-1} f$) is of the form (\ref{D=4}). 

In this section (as well as the previous two) we have therefore shown that (at least up to dimension $D=4$) hypotheses (ii) of Propositions \ref{prop1} and \ref{prop2} are superfluous. To summarize, we have proven the following two propositions.
\begin{prop} \label{prop6}
An arbitrary invertible matrix $M$ of order 2, 3 or 4 can be decomposed as $M = \lambda s$, $\lambda$ being the matrix of a Lorentz transformation and $s$ a symmetric matrix, if and only if the real matrix $\eta M^{t} \eta M $ has a real square root.  
\end{prop}

\begin{prop} \label{prop7}
For space-time dimensions 2, 3 and 4, there exist vielbeins $e_{A}{}^{\mu}$ and $\fv^{B}{}_{\nu}$ corresponding to the metrics $g_{\mu\nu}$ and $f_{\mu\nu}$ respectively (i.e. $\eta^{AB}e_{A}{}^{\mu}e_{B}{}^{\nu}=g^{\mu\nu}$ and $\eta_{AB}\fv^{A}{}_{\mu}\fv^{B}{}_{\nu}=f_{\mu\nu}$) such that $e_{A}{}^{\mu}\fv_{B\mu}=e_{B}{}^{\mu}\fv_{A\mu}$, if and only if there exists  a real matrix $\gamma$ 
such that $\gamma^{\mu}{}_{\rho}\gamma^{\rho}{}_{\nu}=g^{\mu\rho}f_{\rho\nu}$ (i.e. $\gamma^{2}=g^{-1}f$).
\end{prop}

We expect that these results continue to hold in higher dimensions even though we do not have a dimension independent proof. 

\label{4D}

\section{Time evolution and application to ghost-free massive gravity}

Now that we have discussed the different necessary and sufficient conditions for (\ref{CONS}) to hold, we may ask ourselves if these conditions are preserved through time evolution. It is easy to see that there is no general answer to this question i.e. it depends on the theory. Consider for example the case of a bimetric theory where the two metrics are not coupled to each other (or just very weakly). The action of such a theory in four dimensions is given by
\be
S_{ex}=M_{f}^{2}\int d^{4}x \sqrt{-f}R_{f}+M_{g}^{2}\int d^{4}x \sqrt{-g}(R_{g}-2\Lambda)\ .
\ee 
It is easy to see that in some coordinate patch a solution to the equations of motion of this theory is simply given by
\ba
f_{\mu\nu}dx^{\mu}dx^{\nu}&=&dt^{2}-dx^{2}+dy^{2}+dz^{2}\ ,\\
g_{\mu\nu}dx^{\mu}dx^{\nu}&=&-dt^{2}+e^{\sqrt{\Lambda}t}(dx^{2}+dy^{2}+dz^{2})\ .
\ea
This solution corresponds to Minkowski space-time for $f_{\mu\nu}$ and de Sitter space-time for $g_{\mu\nu}$. In particular, whatever the time coordinate $t$ 
\be
g^{-1}f = diag(-1,-e^{-\sqrt{\Lambda}t},e^{-\sqrt{\Lambda}t},e^{-\sqrt{\Lambda}t})\ .
\ee
At $t=0$ this matrix reduces to $diag(-1,-1,1,1)$ and admits a real square root $\gamma$ such that $f\gamma$ is symmetric. For instance
\be
\gamma=\left(
\begin{array}{cccc}
 0 & 1 & 0 & 0 \\
 -1 & 0 & 0 & 0 \\
 0 & 0 & 1 & 0 \\
 0 & 0 & 0 & 1 \\
\end{array}
\right)
\ee
clearly verifies the above conditions. This means that on the $t=0$ hypersurface, one may choose vierbeins obeying condition (\ref{CONS}). However as soon as $t\neq 0$ this condition ceases to be true as $g^{-1}f$ does not even admit a real square root anymore. Thus in the above theory, condition (\ref{CONS}) is not preserved under time evolution. 

In contrast, let us consider the recently proposed dRGT theory  \cite{deRham:2010kj,deRham:2010ik,deRham:2011rn}. We first note that in the metric formulation of this theory, one assumes the existence of a real square root of $g^{-1}f$ (where $g$ is a dynamical metric and $f$ a non-dynamical one); then, according to proposition \ref{prop7}, this mere assumption is equivalent to assuming the existence of vierbeins verifying condition (\ref{CONS}). On the other hand, in the vielbein formulation of dRGT theory, it has been shown in \cite{us} (see also \cite{Hinterbichler:2012cn}) that, at least for some region of parameter space, condition (\ref{CONS}) is imposed by the equations of motion and is therefore preserved under time evolution. When this is the case, the propositions proven in this work then also imply that the existence of the matrix square root of  $g^{-1}f$ is dynamically imposed.

\section{Conclusions}
In this note, we studied in detail the sufficient and necessary conditions for two vielbeins $\fv^{A}$ and $\gv^{B}$ associated 
with two metrics $f_{\mu \nu}$ and $g_{\mu \nu}$ defined on a given manifold to be chosen so that they obey the symmetry condition (\ref{CONS}) which has been used as a gauge condition in vielbein gravity or massive gravity. We also studied as a byproduct the necessary and sufficient condition for an arbitrary matrix $M$ to be decomposed as in (\ref{POLARLOREN}). We showed that, in contrast to what has sometimes been claimed in the literature, the condition (\ref{CONS}) and the decomposition (\ref{POLARLOREN}) cannot be achieved in general but require some extra assumptions related to the existence and properties of square roots of matrices. These assumptions are gathered in Propositions 1 to 7 of the present work. An example where this result is particularly relevant is dRGT massive gravity. Indeed, this theory has been considered in two different frameworks: the first one uses two metrics $f$ and $g$ in such a way that the mass term involves the symmetric polynomials of $\gamma=\sqrt{g^{-1}f}$ \cite{deRham:2010kj,deRham:2010ik,deRham:2011rn,Hassan:2011tf,Hassan:2011zd}, while the second one relies on two vielbeins $E^{A}$ and $L^{B}$ and the mass term is polynomial in these 1-forms\footnote{More precisely it is given by $\sum_{n}\beta_{n}\epsilon_{A_{1}\dots A_{D}}E^{A_{1}}\wedge \dots\wedge E^{A_{n}}\wedge L^{A_{n+1}}\wedge\dots\wedge L^{A_{D}}\ .$} \cite{Hinterbichler:2012cn}. A consequence of our results is that, in general, these two formulations are not equivalent. They become so only when condition (\ref{CONS}) is satisfied. In a region of parameter space it has been shown in \cite{us} that the above condition holds as a consequence of the equations of motion, and thus the equivalence is true dynamically. In the complementary parameter space region however, this is not true in general and it is even possible that the real square-root $\gamma$ does not exist. 

We also showed that, in general, in the 4 dimensional case, it is enough to assume that the matrix $g^{-1}f$ admits a real square root, in order to satisfy a sufficient condition for (\ref{CONS}) to be true. However, for general theories with two metrics, this assumption may be violated dynamically as can be seen explicitly from the example of two decoupled metrics obeying Einstein's equations\footnote{The condition may hold on some initial Cauchy surface and be violated later on.}.


\section*{Acknowledgments}
We thank C.~de Rham, G.~Esposito-Farese, G.~Gababdadze, F.~Hassan, K.~Hinterbichler, S.~Mukhanov, M.~Reuter, R.~Rosen, M.~Sasaki, A.~Tolley, M.~Volkov, R.~Woodard and especially D.~Steer for discussions.

\end{document}